\newif\ifabstr
\abstrfalse

\ifabstr 
\documentclass[runningheads]{llncs}
\usepackage[dvips]{epsfig}
\usepackage{amsmath,amssymb,graphicx}

\else 
\documentclass[11pt]{article}

\usepackage[latin1]{inputenc}

\usepackage{amsmath,amssymb,amsthm,graphicx}
\usepackage[dvips]{epsfig}
\usepackage{amsfonts}
\usepackage[left=1.4in, right=1.4in]{geometry}

\newtheorem{theorem}{Theorem}
\newtheorem{proposition}[theorem]{Proposition}

\newtheorem{question}[theorem]{Question}
\newtheorem{lemma}[theorem]{Lemma}
\theoremstyle{remark}

\theoremstyle{definition}
\newtheorem{definition}[theorem]{Definition}
\fi 

\newcommand{\R}{\mathbb{R}}
\newcommand{\N}{\mathbb{N}}

\newcommand{\C}{\mathcal{C}}

\newcommand{\I}{\mathcal{I}}
\newcommand{\T}{\mathcal{T}}
\newcommand{\K}{{\sf K}}

\newcommand{\heading}[1]{\medskip\par\noindent{\bf #1}}

\ifabstr 

\title{Recognition of Collapsible Complexes is NP-complete}

\author{Martin Tancer}

\institute{
  Institutionen f\"{o}r matematik, Kungliga Tekniska
  H\"{o}gskolan, 100~44 Stockholm. Supported by the G\"{o}ran Gustafsson
  postdoctoral fellowship.\\
  \bigskip
  Department of Applied
  Mathematics, Charles University in Prague, Malostransk\'e n\'am. 25, 118 00,
  Praha 1.\\ 
  \bigskip
  Email: \email{tancer@math.kth.se}
}

\else 

\title{Recognition of collapsible complexes is NP-complete}

\author{
Martin Tancer\thanks{Institutionen f\"{o}r matematik, Kungliga Tekniska
H\"{o}gskolan, 100~44 Stockholm and Department of Applied
Mathematics, Charles University in Prague, Malostransk\'e n\'am. 25, 118 00,
Praha 1. Supported by the G\"{o}ran Gustafsson
postdoctoral fellowship. Partially supported by the project CE-ITI (GA
\v{C}R P202/12/G061) of the Czech Science Foundation.
}}

\date{}

\fi 

\begin{document}

\maketitle


\begin{abstract}
We prove that it is NP-complete to decide whether a given (3-dimensional)
simplicial complex is collapsible. This work extends a result of Malgouyres
and Franc\'{e}s showing that it is NP-complete to decide whether a given
simplicial complex collapses to a $1$-complex.
\end{abstract}


\section{Introduction}

A classical question often considered in algebraic topology is whether some
topological space is contractible. When we consider this question as an
algorithmic question, that is, we consider the topological space as an input for
an algorithm (say as a finite simplicial complex\footnote{Many topological
spaces cannot be represented as finite simplicial complexes. However, we only
consider those that can be represented this way. Already for such spaces, this
question is undecidable.}), then it turns out that
this question is algorithmically undecidable by a result of~Novikov;
see~\cite[\S 10]{volodin-kuznetsov-fomenko74} (see also Appendix).

An important, algorithmically recognizable, subclass of contractible complexes
is the class of collapsible complexes, introduced by Whitehead.
Roughly speaking, a simplicial complex is collapsible if it can be shrunk to a point by a sequence of face collapses, which preserve the homotopy type. (The precise definition is given in the following section.) 

We focus on the computational complexity of the collapsibility problem,
considered as an algorithmic question.
We show that this question
is NP-complete. More precisely, we obtain NP-completeness even if we restrict
the input to $3$-dimensional complexes.

\begin{theorem}
\label{t:main}
It is NP-complete to decide whether a given $3$-dimensional simplicial complex is
collapsible.
\end{theorem}

It is easy to see that this problem belongs to NP (it is just sufficient to
guess a right sequence of elementary collapses), thus the core of our paper
relies on showing the NP-hardness.

By attaching a $d$-simplex to the complexes used in the proof of
Theorem~\ref{t:main} it is easy to observe that Theorem~\ref{t:main} is also valid if we
replace `$3$-dimensional' with `$d$-dimensional' for any $d \geq 4$. We provide
more details in the conclusion (Section~\ref{s:conclusion}).

\heading{Previous work.}
E\u{g}ecio\u{g}lu and Gonzalez~\cite{egecioglu-gonzalez96} have shown that it
is (strongly) 
NP-complete to decide whether a given $2$-dimensional complex can be collapsed
to a point by removing at most $k$ triangles where $k$ is a part of the input. This
problem, however, becomes polynomial-time solvable 
when $k$ is fixed as pointed out by Joswig and Pfetsch~\cite{joswig-pfetsch06} or Malgouyres and Franc\'{e}s~\cite{malgouyres-frances08}. In particular,
deciding whether a $2$-dimensional complex collapses to a point is
polynomial-time solvable. The same approach yields to the fact that deciding
whether a $d$-complex collapses to a $(d-1)$-complex is polynomial time
solvable. Since the author is not aware of a reference, we include a simple
proof here; see Proposition~\ref{p:codim}. 

Malgouyres and Franc\'{e}s~\cite{malgouyres-frances08} have shown that it is
NP-complete to decide whether a given $3$-dimensional complex collapses to some
$1$-complex. Naturally, they asked about the complexity of the problem of
deciding whether a given $3$-complex is collapsible. Theorem~\ref{t:main} 
answers this question in terms of NP-completeness.

Our approach relies, in a significant part, on the work of Malgouyres and
Franc\'{e}s. We sketch their proof as well as point out the differences in
Section~\ref{s:mf}. One of the important differences is that we need to replace
a very simple `clause gadget' of Malgouyres and Franc\'{e}s with something
more suitable for our setting. For this we need to introduce Bing's house with
three rooms, which is done in Section~\ref{s:bings}. Then we construct all 
gadgets needed for our reduction in Section~\ref{s:construction} and we finish
the proof of Theorem~\ref{t:main} in Section~\ref{s:collapse}.

\heading{Links to discrete Morse theory.} Using the result of E\u{g}ecio\u{g}lu and
Gonzalez~\cite{egecioglu-gonzalez96}, Joswig and Pfetsch~\cite{joswig-pfetsch06}
proved that it is (strongly) NP-complete to decide whether there exists a Morse
matching with at most $c$ critical cells where $c$ is a part of the input. 
Our main result can be reformulated in terms of Morse matchings in the
following way.

\begin{theorem}
\label{t:morse}
It is NP-complete to decide whether a given $3$-dimensional simplicial complex
admits a perfect Morse matching. 
\end{theorem}

In the short proof below we keep a few notions undefined. We refer to Joswig
and Pfetsch~\cite{joswig-pfetsch06} for these notions as well as for further details on Morse matchings in the
computational complexity context.

\begin{proof}
We use the fact that a simplicial complex $K$ collapses to a point $x$
if and only if it admits a Morse function such that $K \setminus \{x, 
\emptyset\}$ contains no critical cells (see Forman~\cite{forman98}). 
Therefore, $K$ is collapsible if and only if it admits a perfect Morse matching
(see the exposition by Joswig and Pfetsch~\cite{joswig-pfetsch06}).
Consequently, Theorem~\ref{t:morse} is equivalent to Theorem~\ref{t:main}.
\end{proof}

Efficient search for collapsing sequences also plays an important role for
searching a CW-complex homotopy equivalent to a given simplicial complex with
number of critical cells as small as possible. This has further impact on
efficient homology computations. See~\cite{benedetti-lutz14}.

\heading{Links to shape-reconstruction.} The task of shape reconstruction is to
reconstruct a shape from a set of points that sample it. An important
subtask is to reconstruct the homotopy type or the homeomorphism type of the shape. In a recent work of
Attali and Lieutier~\cite{attali-lieutier_online_first15}, the aim is to collapse
the Rips complex or the \v{C}ech complex of the sampling set to a complex
homeomorphic with the shape. In this context, Theorem~\ref{t:main} means certain
limitations for results one can expect. In particular, specific treatment
in~\cite{attali-lieutier_online_first15}, using the 
properties of the Rips and \v{C}ech complexes, seems important.

\heading{Another notion of collapsibility.} In the context of Helly-type theorems
in discrete geometry, Wegner~\cite{wegner75} introduced a notion of
$d$-collapsibility. This notion shares some properties with collapsibility, but
for example, it does not preserve the homotopy type. The author has shown
in~\cite{tancer10np} that the recognition of $d$-collapsible complexes is NP-hard
for $d \geq 4$. We remark that the approach in that case is different and
the result from~\cite{tancer10np} should not be confused with the result
presented here on classical (Whitehead's) collapsibility.

\section{Preliminaries}

\heading{Simplicial complexes.}
We work with finite \emph{(abstract) simplicial complexes}, that is, with
set systems $K \subseteq 2^V$ such that $V$ is a finite set and 
if $\alpha \in K$ and $\beta \subset \alpha$, then $\beta \in K$. We recall few
basic definitions; however, we also assume
that the reader is familiar with some basic properties of simplicial complexes.
Otherwise we refer to any of the books~\cite{hatcher01,matousek03,munkres84}.
In particular, we assume that the reader is familiar with the correspondence of
abstract simplicial complexes and geometric simplicial complexes since it will
be very convenient in the further text to define some simplicial complexes by
pictures.

Elements of a simplicial complex $K$ are \emph{faces} (or \emph{simplices}).
A \emph{$k$-face} is a face of \emph{dimension} $k$, that is, a face in $K$ of
size $k + 1$. $0$-dimensional, $1$-dimensional, and $2$-dimensional faces are
\emph{vertices}, \emph{edges}, and \emph{triangles} respectively. 

When we consider a simplicial complex as an input for an algorithm, it is given
by a list of all faces.

\heading{Collapsibility.}
Let $\sigma$ be a nonempty non-maximal face of $K$. We say that $\sigma$ is \emph{free}
if it is contained in only one maximal face $\tau$ of $K$. 
Let $K'$ be the simplicial complex obtained from $K$ by removing $\sigma$ and
all faces above $\sigma$, that is,
$$
K' := K \setminus \{\vartheta \in K\colon \sigma \subseteq \vartheta \}.
$$
We say that $K'$ arises from $K$ by an \emph{elementary collapse} (induced by $\sigma$ 
and $\tau$). 
We say that a complex $K$ \emph{collapses} to
a complex $L$ if there exist a sequence of complexes $(K_1 = K, K_2, \dots, K_{m-1},
K_m = L)$, called a \emph{sequence of elementary collapses} (from $K$ to $L$), such that $K_{i+1}$ arises 
from $K_i$ by an elementary collapse for any $i \in
\{1, \dots, m-1\}$. A simplicial complex $K$ is \emph{collapsible} if it collapses to
a point. 

Let $(K_1 = K, K_2, \dots, K_{m-1}, K_m = L)$ be a sequence of elementary
collapses. Then for
every $\eta \in K \setminus L$ there is a unique complex $K_i$ such that $\eta
\in K_i$ and $\eta \not \in K_{i+1}$. Then we say that $\eta$ \emph{collapses}
in this step. In particular, we will often use phrases such as `$\eta_1$
collapses before $\eta_2$'.

\ifabstr
\else

\heading{Collapsibility with constrains.}
In our constructions, we will often encounter the following situation: We will
be given a complex $L$ glued to some other complexes forming a complex $M$. We
will know some collapsing sequence of $L$ and we will want to use this
collapsing sequence for $M$. This might or might not be possible. We will set
up a sufficient condition.

\begin{definition}
Let $M$ be a simplicial complex and $L$ be a subcomplex of $M$. We define the
\emph{constrain complex} of pair $(M,L)$ as
$$
\Gamma = \Gamma(M,L) := \{\vartheta \in L\colon \vartheta \subseteq \eta \hbox{ for some
} \eta \in M \setminus L\}.
$$
\end{definition}

The constrain complex is obviously a subcomplex of $L$. Now we can present
an elementary condition when collapsing of $L$ induces collapsing of $M$.

\begin{lemma}
\label{l:constrain}
Let $M$ be a complex, $L$ subcomplex of $M$ and $\Gamma$ be the constrain complex of
$(M,L)$. We also assume that $L$ collapses to $L'$ containing $\Gamma$. Then $M$
collapses to
$$M' := L' \cup (M \setminus L).$$
\end{lemma}

\begin{proof}
Let $(L_1 = L, L_2, \dots, L_{m-1}, L_m = L')$ be a sequence of elementary
collapses. Let $\sigma_i$ be the face of $L_i$ which is collapsed in order to
obtain $L_{i+1}$ and let $\tau_i$ be the unique maximal face in $L_i$
containing $\sigma_i$. We also set $M_i = L_i \cup (M \setminus L)$. 
The assumption ensures us that all superfaces of $\sigma_i$ in $M_i$ belong
to $L_i$. Therefore $(M_1 = M, M_2, \dots, M_{m-1}, M_m = M')$ is a sequence of
elementary collapses still induced by $\sigma_i$ and $\tau_i$. 
\end{proof}
\fi

\heading{Collapsibility in codimension 1.}
Here we show that collapsibility in codimension 1 is polynomial-time solvable.
This result is not needed for the proof of Theorem~\ref{t:main}; it only serves
as a complementary result.  We need the proposition below. 
The proposition implies that 
we can
collapse an input $d$-complex $K$ greedily, and with this
greedy algorithm, we
obtain a $(d-1)$-complex if and only if $K$ collapses to a $(d-1)$-complex $L$.
\begin{proposition}
\label{p:codim}
Let $K$ be a $d$-complex which collapses to a $(d-1)$-complex $L$ and to some
$d$-complex $M$. Then $M$ collapses to a $(d-1)$-complex.
\end{proposition}
\begin{proof}
Let $(K_1 = K, K_2, \dots, K_{m-1}, K_m = L)$ be a sequence of elementary
collapses where the collapse from $K_i$ to $K_{i+1}$ is induced by faces $\sigma_i$ and
$\tau_i$. Note that every $d$-dimensional face of $K$ is $\tau_i$ for some $i$.

Let $j$ be the smallest index such that $\tau_j$ belongs to $M$ and let $\eta_j$ be 
a $(d-1)$-face with $\sigma_j \subseteq \eta_j \subset \tau_j$. Note also that
no $\tau_i$ with $i < j$ belongs to $M$ due to our choice of $j$. However,
since $\sigma_j$ is free in $K_j$ and therefore $\eta_j$ is free as well, the
only $d$-faces of $K$ containing $\eta_j$ might be the faces $\tau_i$ with $i
\leq j$. Altogether, $\tau_j$ is the unique $d$-face of $M$ containing $\eta_j$
and we can collapse $\eta_j$ (removing $\tau_j$). If $M$ is still
$d$-dimensional, we repeat our procedure. After finitely many steps we obtain a
$(d-1)$-complex.  
\end{proof}

\section{The approach by Malgouyres and Franc\'{e}s}
\label{s:mf}

In this section we describe the approach of Malgouyres and
Franc\'{e}s~\cite{malgouyres-frances08}. In
some steps we follow their approach almost exactly; however, there are also
steps that have to be significantly 
modified in order to obtain our result. We will emphasize the steps where our approach differs.

The reduction is done, as usual, from $3$-satisfiability problem which is well
known to be NP-complete. We assume that the reader is familiar with the related
terminology.
Given a $3$-CNF formula\footnote{That is, a formula in conjunctive normal form.} $\Phi$, Malgouyres and Franc\'{e}s construct a
$3$-di\-men\-sion\-al complex $\C(\Phi)$ such that $\C(\Phi)$ collapses to a
$1$-complex if and only if $\Phi$ is satisfiable. They compose $\C(\Phi)$ of
several smaller complexes that we will call \emph{gadgets}. For every literal
$\ell$ in the formula they introduce a \emph{literal gadget} $\C(\ell)$.
(This includes introducing $\C(\bar \ell)$ where $\bar \ell$ is the negation of
$\ell$.) The gadgets $\C(\ell)$
and $\C(\bar \ell)$ are glued along an edge so that a major part of only
$\C(\ell)$ or only $\C(\bar \ell)$ can be collapsed in the first phase of
collapsing. Another gadget is a \emph{conjunction gadget} $\C_{and}$ glued to literal
gadgets via \emph{clause gadgets} so that $\C_{and}$ can be collapsed at some
step if and only if every clause contains a literal $\ell$ such that
the major part of $\C(\ell)$ was already collapsed, that is, if and only if
$\Phi$ is satisfiable. As soon as $\C_{and}$ is
collapsed, it makes few other faces of the literal gadgets free which enables
to collapse the
whole complex to a $1$-dimensional complex. As it follows from the
construction of Malgouyres and Franc\'{e}s, if the formula is satisfiable, 
the resulting $1$-complex contains many cycles and therefore it cannot be further collapsed to a point.

Our idea relies on filling the cycles of the resulting $1$-complex so that we
can further proceed with collapsings. However, we cannot fill the cycles
completely naively, since we do not know in advance which $1$-complex we obtain.
In addition filling these cycles naively could possibly introduce new
collapsing sequences starting with edges on the boundaries of the filled cycles
which could possibly yield to collapsing the complex even if the formula were
not satisfiable. Therefore, we have to be very careful with our 
\ifabstr
\else  
construction (which unfortunately means introducing few more technical
steps).\fi

We are going to construct a simplicial complex $K(\Phi)$ such that $K(\Phi)$ collapses to a point if and only if
$\Phi$ is satisfiable. 
In fact, our complex $K(\Phi)$ will always be
contractible, independently of $\Phi$ (although we do not need this fact in our
reduction; and therefore we do not prove it). 
We reuse literal and conjunction gadgets of
Malgouyres and Franc\'{e}s (only with minor modifications regarding
distinguished subgraphs). 
Unfortunately, we need to replace the very simple clause
gadget of Malgouyres and Franc\'{e}s (it consists of two triangles sharing an
edge or is even simpler, depending on the clause). For this we need to introduce Bing's house with three
rooms and three thick walls. We also need \emph{disk gadgets} which fill the
cycles in the resulting $1$-complex. We remark that the disk gadgets will not be
topological disks but only some contractible complexes. However, we keep the
name disk gadgets because of the idea of filling the cycles. 

\section{Bing's rooms and Bing's house with three rooms}
\label{s:bings}

In our reduction we will need several auxiliary constructions that we suitably
glue together. We present them in this section.

\heading{Bing's rooms.} We will consider Bing's house as a simplicial
complex obtained by gluing two smaller simplicial complexes called Bing's
rooms. Later we will use these rooms for building more complicated Bing's
house with three rooms. \emph{Bing's room with a thin wall} is a complex
depicted in Figure~\ref{f:rooms} on the left and \emph{Bing's room with a thick
wall} is in the middle. The room with a thin wall contains only $2$-dimensional
faces whereas the room with a thick wall contains one $3$-dimensional block
obtained by thickening one of the walls. Both rooms contains two holes in the ground
floor and one hole in the roof. If, starting from Bing's house with thick wall,
we collapse away the thick wall, we obtain a complex that we call
\emph{Bing's room with collapsed thick wall}, shown on the right. 
(Note that the left bottom edge of the collapsed thick wall is
still present although it is not contained in any 2-dimensional face.)

\begin{figure}[t]
\begin{center}
\includegraphics{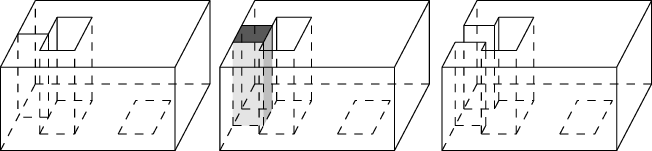}
\caption{Bing's room with a thin wall (on the left), Bing's room with a
thick wall (in the middle), and Bing's room with collapsed thick wall (on the
right).}
\label{f:rooms}
\end{center}
\end{figure}  

\heading{Bing's house with \ifabstr\else one thin and one
thick wall.\fi} If we rotate the
ground floor
of one of the rooms and we glue the two rooms together along the ground floor, we
obtain \emph{Bing's house with \ifabstr
\else one thin and one
thick wall\fi} as introduced by
Malgouyres and Franc\'{e}s~\cite{malgouyres-frances08}. See
Figure~\ifabstr\else\ref{f:thin_thick}\fi. 
\ifabstr
\else
\begin{figure}
\begin{center}
\includegraphics{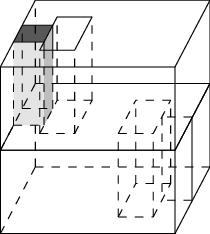}
\caption{Bing's house with one thin and one thick wall.}
\label{f:thin_thick}
\end{center}
\end{figure}  
\fi
Similarly, we can obtain Bing's house with two thin walls or Bing's house
\ifabstr\else with two
thick walls.\fi

\heading{Bing's house with three rooms and three thick walls.} As an auxiliary construction, we
also need to introduce Bing's house with three rooms. First we consider
the \emph{base floor} depicted 
\ifabstr
\else
in~Figure~\ref{f:template}.
\fi 
It consists of three
quadrilaterals with holes, glued together. For simplicity of explanations, we
will assume that all these three quadrilaterals are squares. Also, we assume
that that the holes are squares. Now we consider three Bing's rooms with thick walls labeled $1$,
$2$ and $3$. The Bing room with label $i$ is glued to the two squares with
label $i$ so that the grey part of one of the squares with label $i$ is the
place where the thick wall of the room is glued to the base floor. Here, it is
important that we do not have to distinguish whether the rooms are glued to the
base floor from below or from above, since we could not place them in such a way
(in 3D) simultaneously.\footnote{It can be shown that the resulting complex
does not topologically embed into $\R^3$ since it contains a
M\"{o}bius band and an annulus glued to the central cycle of the band along
one of the boundary components of the annulus. (The annulus is one third of the
base floor and the band is formed by parts of the outer walls of the three
rooms.) However, this non-embeddability fact is far beyond the
needs of this paper.} 
The resulting complex we call \emph{Bing's house with three rooms (and
three thick walls)}. We remark that Bing's house with three rooms is
contractible which can be shown in a similar way as contractibility of
classical Bing's house. (If glue cuboid bricks to the base floor instead of
Bing's rooms, we obtain a complex $L$ which is obviously contractible. Bing's house
with three rooms is obtained by `digging holes' into $L$.) Later on, we will need a specific collapsing sequence of
Bing's house with three rooms. The existence of such a sequence implies
contractibility as well.

\ifabstr
\else
\begin{figure}
\begin{center}
\includegraphics{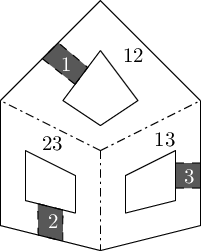}
\caption{The base floor of Bing's house with three rooms and with three thick
walls.}
\label{f:template}
\end{center}
\end{figure}
\fi

\heading{Bing's house with three collapsed walls.} For further purposes it will
be convenient to work with Bing's house with three rooms where the thick 
walls are collapsed. We let each of the thick walls collapse to the edge on the
base floor. This way we obtain \emph{Bing's house with three
collapsed walls}. We provide the reader with a drawing with two rooms
only (this can still be done in three dimensions); 
\ifabstr
\else
see Figure~\ref{f:coll_walls}. 
\fi
We also distinguished edges $x_1$, $x_2$ and $x_3$
such that $x_i$ is the only remaining edge of the thick wall in room $i$ after
collapsing the wall. Note that $x_1$, $x_2$ and $x_3$ are the only free faces
of Bing's house with three collapsed walls. See also Figure~\ref{f:clause_g}
for the base floor.

\ifabstr\else
\begin{figure}
\begin{center}
\includegraphics{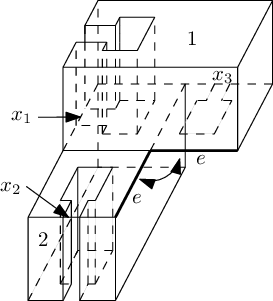}
\caption{Two blocks of Bing's house with three collapsed walls. The edges
marked with $e$ are glued together.}
\label{f:coll_walls}
\end{center}
\end{figure}
\fi

\heading{Triangulations.}
In order to obtain simplicial complexes we need to triangulate our gadgets that
we obtain from Bing's rooms, Bing's houses etc. It will not be important for
us how do we precisely triangulate pieces in the construction of dimensions $2$
or less. \ifabstr
\else
For example, the middle level of Bing's house with one thick and one
thin wall can be triangulated as suggested in Figure~\ref{f:triangulating}
keeping in mind that the triangulations of particular $2$-cells have to be
compatible on the intersections. In some cases, we will need gadgets with many
prescribed edges in some part of the triangulation where the number of these edges depends on the size of the 
$3$-CNF formula we will work with (see Figure~\ref{f:conj_g}~or~\ref{f:bl}). In
such cases we require that the size of the triangulation is polynomial in the
number of prescribed edges.

\begin{figure}
\begin{center}
\includegraphics{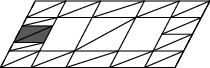}
\caption{A suitable triangulation of the middle level of Bing's house with one
thick and one thin wall.}
\label{f:triangulating}
\end{center}
\end{figure}
\fi

The only $3$-cells appearing in our construction are thick walls of Bing's rooms
(houses). For these thick walls we use particular triangulations of Malgouyres
and Franc\'{e}s~\cite{malgouyres-frances08}. 
\ifabstr
\else
The thick wall is subdivided into
four prisms $012389$, $014589$, $236789$ and $456789$. See
Figure~\ref{f:triang_thick}, left. Each prism is further subdivided into two
simplices (which are not shown on the picture). 
This triangulation allows
collapsing the thick wall into two smaller complexes from
Figure~\ref{f:triang_thick}, middle and right (in the middle picture, the edge
$01$ is contained in no $2$-cell; the right picture is drawn from behind and the edge $89$ is contained
in no $2$-cell). In particular, the collapsing from the middle picture is used
when obtaining Bing's room with a collapsed thick wall from Bing's room with a
thick wall.

\begin{figure}
\begin{center}
\includegraphics{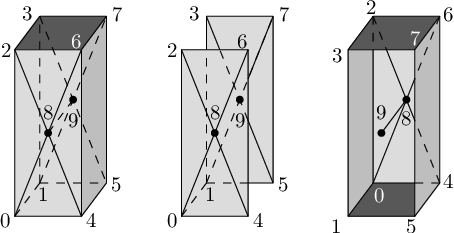}
\caption{Collapsings of the thick wall. (Usually, we use gray only for
$3$-cells or attachments of $3$-cells. In this case, the objects in the middle
and on the right are two dimensional while grey emphasize which $2$-cells are
still in the object.)}
\label{f:triang_thick}
\end{center}
\end{figure}
\fi

\section{Construction of the reduction gadgets}
\label{s:construction}
Here we start filling in details of the construction sketched in
Section~\ref{s:mf}.

Given a $3$-CNF formula $\Phi$ we construct a three-dimensional 
simplicial complex $K(\Phi)$ such that $K(\Phi)$ is collapsible if and only if
$\Phi$ is satisfiable. We assume that every clause of $\Phi$ contains exactly
three literals and also that no clause contains a literal and its negation.
\ifabstr
\else

\fi
The complex $K(\Phi)$ will consist of several gadgets described below.
\ifabstr
\else
For
each of the gadgets we also need to find some suitable collapsing sequence. We
usually postpone the proofs that such sequences exist to Section~\ref{s:col_seq}
so that the main idea can be explained while the technical details are left to
the end.
\fi

\heading{Literal gadget.}
First we establish the \emph{literal gadget} $K(\ell, \bar \ell)$ for every pair of literals
$\ell$ and $\bar \ell$. This gadget is by Malgouyres and
Franc\'{e}s~\cite{malgouyres-frances08}, we only glue it to other gadgets in a
different way. It consists of two smaller gadgets $X(\ell)$
and $X(\bar \ell)$ suitably glued together. 

We set $X(\ell)$ to be Bing's house with two thick walls as
in~Figure~\ref{f:literal_g}. It contains two
distinguished edges $e(\ell)$ and $f(\ell)$. Furthermore, it contains a
distinguished path $p(\ell)$
joining the common vertex of $e(\ell)$ and $f(\ell)$ 
with the upper thick wall (this path
contains neither $e(\ell)$ nor $f(\ell)$). Let us emphasize that in this case, we use particular
triangulation by Malgouyres and Franc\'{e}s~\cite{malgouyres-frances08} that
subdivides the upper thick wall into four prisms which are further
triangulated. The path $p(\ell)$ enters the upper wall in vertex $0$ of this
triangulation and it continues to vertex $8$. For $\bar \ell$ we construct
$X(\bar \ell)$ analogously.

The complex $K(\ell, \bar \ell)$ is composed of $X(\ell)$ and $X(\bar \ell)$
glued together along edge $89$. The common vertex $8$ will be important for
further constructions; and therefore we rename it to $u_{\ell, \bar \ell}$
emphasizing dependency on $\ell$ and $\bar \ell$. The following lemma describes
a particular sequence of collapsing the literal gadget that we will use later.
\ifabstr\else
It also says that at least one of the edges $f(\ell)$, $f(\bar \ell)$ has to be
collapsed before collapsing the literal gadget to a $2$-complex.
\fi

\begin{figure}[t]
\begin{center}
\includegraphics{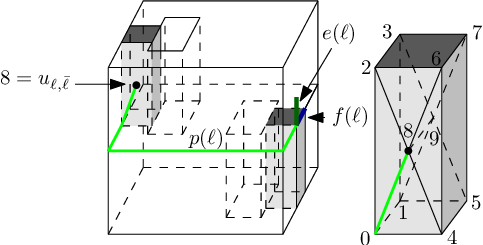}
\caption{The complex $X(\ell)$ from the literal gadget with distinguished edges and
paths. The detailed triangulation of the upper thick wall is on the right.}
\label{f:literal_g}
\end{center}
\end{figure}  

\begin{lemma}
\label{l:literal}
\noindent
\begin{enumerate}
\item
$K(\ell, \bar \ell)$ collapses to a complex that contains only path $p(\ell)$
and edges $e(\ell)$ and $f(\ell)$ from $X(\ell)$ while it contains almost all
$X(\bar \ell)$ with exception that the upper thick wall of $X(\ell)$ 
was collapsed to a thin wall keeping only rectangles $0462$, $0451$, $4576$,
$2673$, and $0132$.\footnote{For simplicity of notation, we keep the same
numbers of vertices either for upper thick wall of $X(\ell)$ or of $X(\bar
\ell)$. However, we once more emphasize that that these two walls share the
edge $89$ only.} \ifabstr\else(Consult Figure~\ref{f:triang_thick} right, if
you remove the edge $89$.)\fi  
\ The~role of
$\ell$ and $\bar \ell$ can be interchanged.
\item 
Let $L(\ell, \bar \ell)$ be the complex resulting in item 1 without the
edge $e(\ell)$. This complex further collapses to the union of the paths
$p(\ell)$, $p(\bar \ell)$ and the edge $e(\bar \ell)$.
\ifabstr\else
\item{Let $T(\ell)$ be any of the two triangles containing $e(\ell)$ and
$T(\bar \ell)$ be any triangle containing $e(\bar \ell)$. Before collapsing
both $T(\ell)$ and $T(\bar \ell)$, at least one of the edges $f(\ell)$, $f(\bar
\ell)$ must be collapsed.}
\fi
\end{enumerate}

\end{lemma}

\ifabstr
\else
\begin{proof}
We postpone the proof of items 1 and 2 on precise collapsing sequences
to Section~\ref{s:col_seq}. Item~3 is already proved by Malgouyres and
Franc\'{e}s; see Remark~1, Example~3 and the 
proof of Theorem~4 in~\cite{malgouyres-frances08}. We sketch here that if neither $f(\ell)$ nor $f(\bar \ell)$
is collapsed, then the only one of the two upper thick walls, one of $X(\ell)$
and one of $X(\bar \ell)$, can be collapsed so that its edge $01$ becomes free. 
Hence, only one of the triangles $T(\ell)$ and $T(\bar \ell)$ might become
free before collapsing $f(\ell)$ or $f(\bar \ell)$.
\end{proof}
\fi

\heading{Conjunction gadget.}
Next we define the \emph{conjunction gadget} $K_{and}$. 
It is Bing's house with one collapsed thick wall and one thin wall. See~Figure~\ref{f:conj_g}
on the left. We also distinguish several edges and vertices of the gadget.

 As an auxiliary construction, for every pair $\ell$, $\bar \ell$ of literals,
we create an anchor-shaped tree $A(\ell, \bar \ell)$ formed of $u_{\ell, \bar \ell}$,
$p(\ell)$, $p(\bar \ell)$, $f(\ell)$ and $f(\bar \ell)$ from $K(\ell, \bar
\ell)$ and furthermore of newly introduced edge $a(\ell, \bar \ell)$ and vertex
$v_{and}$. See~Figure~\ref{f:conj_g} on the right. We glue all trees $A(\ell,
\bar \ell)$ in vertex $v_{and}$ obtaining a tree $A$.

Finally, we let $e_{and}$ to denote the only free edge of
$K_{and}$ and we glue $A$ to the lower left wall of $K_{and}$ as
on~Figure~\ref{f:conj_g} on the left. Note that, in particular, every literal
gadget is glued to the conjunction gadget.


\begin{figure}[t]
\begin{center}
\includegraphics{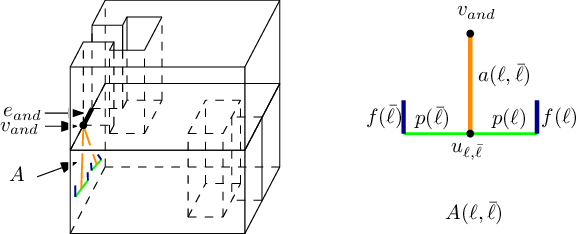}
\caption{Conjunction gadget $K_{and}$ with distinguished edges and paths.}
\label{f:conj_g}
\end{center}
\end{figure}  

After we introduce the remaining gadgets, we will see that $K_{and}$ is glued to
other gadgets only along $A$ and $e_{end}$. 
\ifabstr
\else
The following lemma states
that if we want to collapse $K_{and}$ at some phase of collapsing, we have to
make $e_{end}$ free first in whole $K(\Phi)$ and only then we can continue with
collapsing $K_{and}$. On the other hand, as soon as we make $e_{and}$ free, we
can collapse the complex to $A$. 
\fi

\begin{lemma}
\label{l:and}
\ifabstr
\else
\noindent
\begin{enumerate}
\item
$K_{and}$ collapses to $A$.
\item
Before collapsing any triangle containing one of the edges $f(\ell, \bar \ell)$,
the edge $e_{and}$ has to be collapsed. 
\end{enumerate}
\fi
\end{lemma}

\ifabstr
\else
\begin{proof}
We prove item 1 in Section~\ref{s:col_seq}.
Item 2 of the lemma is explained in~\cite[Remark 1]{malgouyres-frances08}. 
\end{proof}
\fi

\heading{Clause gadget.}
We proceed with introducing the clause gadget. 
For a clause $c
= (\ell_1 \vee \ell_2 \vee \ell_3)$ we set $K(c)$ to be Bing's house with three
collapsed walls as described in Section~\ref{s:bings} and with several
distinguished edges and paths; see Figure~\ref{f:clause_g}. Namely, the only
three free edges of $K(c)$ are labeled
$(\ell_i,c)$. We also distinguish three paths $p(\ell_i,c)$ connecting 
the center of the base floor with $(\ell_i,c)$ (we assume that
$(\ell_i, c)$ is not contained in the path). We also distinguish one other edge
emanating from the center inside the base floor and we label it by
$e_{and}$. This last edge $e_{and}$ is glued together with the edge of the
conjunction gadget labelled $e_{and}$ so that the central vertex of the
base floor becomes the vertex $v_{and}$ of the conjunction gadget.

\begin{figure}[t]
\begin{center}
\includegraphics{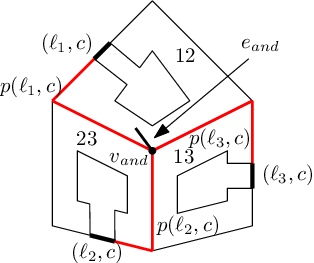}
\caption{The base floor of the clause gadget $K(c)$ with distinguished edges and paths.}
\label{f:clause_g}
\end{center}
\end{figure}  

\begin{lemma}
\label{l:clause}
\ifabstr
\else
\noindent
\begin{enumerate}
\item
$K(c)$ collapses to a complex composed of $e_{and}$, three paths
$p(\ell_i,c)$ and two of the three edges $(\ell_i,c)$.
\item
Any collapsing of $K(c)$ starts with one of the edges $(\ell_i,c)$.
\end{enumerate}
\fi
\end{lemma}
\ifabstr
\else
\begin{proof}
We again postpone the proof of item 1 to Section~\ref{s:col_seq}. Item 2 is
obvious as soon as we realize that the only free faces of $K(c)$ are
the three edges $(\ell_i,c)$.
\end{proof}
\fi

\heading{Disk gadgets.} Finally, for every pair of literals $\ell$, $\bar \ell$
we construct the \emph{disk gadget} $D(\ell, \bar \ell)$ filling empty cycles in the construction of
Malgouyres and Franc\'{e}s~\cite{malgouyres-frances08}. 
As we mentioned before, these gadgets will not be topological disks. However, they are contractible
and play a similar role as disks.


We start with Bing's house with one collapsed thick wall and one thin wall; 
\ifabstr
\else
see Figure~\ref{f:bl}. 
\fi
We
label the only free face of this complex with $e(\ell)$ and glue it to the edge
$e(\ell)$ of $K(\ell, \bar \ell)$. We pick a vertex on the 
edge connecting the left and the bottom wall and label it $v_{and}$. We also glue
this vertex to $v_{and}$ vertex of the conjunction gadget. The edge connecting
$v_{and}$ and one of the vertices of $e(\ell)$ is labelled by $b(\ell)$. Next, for
every clause $c_j$ containing the literal $\ell$ we make a copy of path
$p(\ell,c_j)$ and edge $(\ell,c_j)$ (where the template comes from $K(c_j)$)
starting in vertex $v_{and}$. In particular, $B(\ell)$ is glued to the
complexes $K(c_j)$ along these paths and edges. The resulting complex is denoted by $B(\ell)$. 
We perform an analogous construction for $B(\bar \ell)$.

\ifabstr
\else
\begin{figure}
\begin{center}
\includegraphics{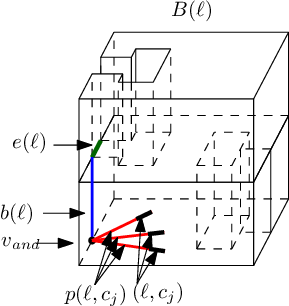}
\caption{Complex $B(\ell)$.}
\label{f:bl}
\end{center}
\end{figure}
\fi

This complex  can be collapsed (inside whole $K(\Phi)$) as soon as the edge
$e(\ell)$ is free. Then it collapses to a complex composed of the distinguished
edges and paths, as the following lemma summarizes.

\begin{lemma}
\label{l:bl}
\ifabstr
\else
\noindent
\begin{enumerate}
\item
$B(\ell)$ collapses to the 1-complex composed of $b(\ell)$, paths $p(\ell,
c_j)$ and edges $(\ell, c_j)$.
\item
Any collapsing of $B(\ell)$ starts with the edge labelled $e(\ell)$.
\end{enumerate}
\fi
\end{lemma}

\ifabstr
\else
\begin{proof} 
As usual, item 1 is proved in Section~\ref{s:col_seq}.

Item 2 is true since $e(\ell)$ is the only free edge of $B(\ell)$.
\end{proof}
\fi

Now we can finally construct $D(\ell, \bar \ell)$; 
\ifabstr
\else
see Figure~\ref{f:d_glue}.
\fi
We fill two cycles with a disk. The first cycle is formed by $b(\ell)$,
$p(\ell)$ and $a(\ell, \bar \ell)$, the second cycle by $b(\bar \ell)$, $p(\bar
\ell)$ and $a(\ell, \bar \ell)$. This finishes the construction of $D(\ell, \bar
\ell)$ and since we have already described all gluings, it also finishes the
construction of $K(\Phi)$.   

\ifabstr\else
\begin{figure}
\begin{center}
\includegraphics{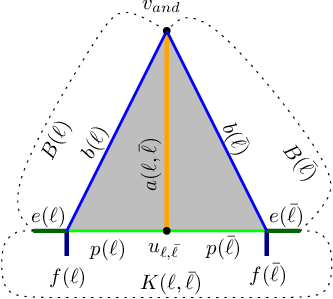}
\caption{Disk gadget $D(\ell,\bar \ell)$ and its attachment to other gadgets.}
\label{f:d_glue}
\end{center}
\end{figure}
\fi

\ifabstr
\else

\section{Collapsibility of $K(\Phi)$}
\label{s:collapse}

In this section we prove that $K(\Phi)$ is collapsible if and only if $\Phi$ is
satisfiable. Thereby we prove Theorem~\ref{t:main}.

\heading{Satisfiable formulas.}
Let us first assume that $\Phi$ is satisfiable and fix one satisfying
assignment of $\Phi$. We construct a collapsing sequence for $K(\Phi)$. We
proceed in several steps (each step will still consist of many elementary
collapses). By $K^{(i)}(\Phi)$ we denote the complex obtained after performing
$i$th step of collapsing. We use similar notation for gadgets, for example,
$K_{and}^{(i)}$ is the remaining part of $K_{and}$ after $i$th step, that is, $K^{(i)}(\Phi) \cap
K_{and}$.

\heading{Step 1.} For every literal $\ell$ we start with ``partial'' collapsing of $K(\ell, \bar
\ell)$ such as in Lemma~\ref{l:literal}~(1). Note that that the constrain
complex of the pair $(K(\Phi), K(\ell, \bar \ell))$ consist of $p(\ell)$, $p(\bar
\ell)$, $e(\ell)$, $e(\bar \ell)$, $f(\ell)$ and $f(\bar \ell)$; therefore Lemma~\ref{l:constrain} induces
collapsing on whole $K(\Phi)$. 
If $\ell$ has positive occurrence in
the assignment, we let $X(\ell)$ collapse to $p(\ell)$, $e(\ell)$ and $f(\ell)$ while
in $X(\bar \ell)$ only the upper thick wall of $X(\bar \ell)$ collapses to a thin wall. This collapsing
makes the edge $e(\ell)$ free. 

We gradually perform this collapsing for all literals with positive occurrence.
Note that by considering positive occurrences only, we do not ``miss'' negative
ones since for every variable $u$ exactly one literal among $u$ and $\neg u$ has positive occurrence.

\heading{Step 2.} We continue with collapsing $B(\ell)$ as stated in Lemma~\ref{l:bl}~(1).
Observe that at this stage, the constrain complex of the pair $(K^{(1)}(\Phi),
B^{(1)}(\ell)) = (K^{(1)}(\Phi), B(\ell))$ contains only 
$b(\ell)$, paths $p(\ell, c_j)$ and edges $(\ell, c_j)$ (in particular the edge
$e(\ell)$ is not in it as well as the vertex of $e(\ell)$ which is not adjacent
to $b(\ell)$), therefore collapsing
from Lemma~\ref{l:bl} induces collapsing of $K^{(1)}(\Phi)$ by
Lemma~\ref{l:constrain}.

In further text we will use Lemma~\ref{l:constrain} many times in a similar fashion
without mentioning it explicitly. (We will describe the constrain complex only.)

\heading{Step 3.}
Now, since the assignment is satisfying for every clause $c$ at least one of
the edges $(\ell_i, c)$ became free. Therefore, every clause gadget collapses
to a $1$-complex described in Lemma~\ref{l:clause}~(1). The constrain complex
for the pair $(K^{(2)}(\Phi), K^{(2)}(c))$ is a subcomplex of the complex
formed by paths $p(\ell_j, c)$ and edges $(\ell_j,c)$ with $j \neq i$.

\heading{Step 4.}
Now we focus on the edge $e_{and}$. At the beginning, it was contained in
triangles in clause gadgets and in a single triangle of the conjunction gadget.
All triangles of clause gadgets were collapsed, therefore $e_{and}$ is free
now. According to Lemma~\ref{l:and}~(1), we can collapse the conjunction
gadget $K_{and}$ to $A$ now (checking that the constrain complex for
$(K^{(3)}(\Phi),K^{(3)}_{and})$ is $A$).

\heading{Step 5.}
In this step, we will collapse the literal and the disk gadgets. The important
fact is that the edges $f(\ell)$ and $f(\bar \ell)$ are already free.
Therefore, we can
proceed with collapsing $K^{(4)}(\ell,\bar \ell)$ according to
Lemma~\ref{l:literal}~(2). This leaves $e(\bar \ell)$ free as well as all
remaining edges of paths $p(\ell)$ and $p(\bar \ell)$. Now we can easily
collapse $B(\bar \ell)$ according to Lemma~\ref{l:bl}~(1) and consequently also
the $D(\ell, \bar \ell)$ (having all boundary edges free). 

\heading{Step 6.}
Now we have a collection of paths emanating from $v_{and}$ (which are remainders of clause gadgets). This collection can be easily
collapsed to a point, say $v_{and}$.

\heading{Non-satisfiable formulas.}
Now we show that $K(\Phi)$ is not collapsible for non-satisfiable formulas.
More precisely, we assume that $K(\Phi)$ is collapsible and we deduce that
$\Phi$ is satisfiable.

If $K(\Phi)$ is collapsible, then in particular some triangle of
$K_{and}$ has to be collapsed. We investigate what
had to be collapsed before collapsing a first triangle of $K_{and}$.
According
to Lemma~\ref{l:and}~(2) the edge $e_{and}$ has to be made free before 
this step. Lemma~\ref{l:clause}~(2) implies that for every clause $c$ there is
a literal $\ell(c)$ in this clause such that the edge $(\ell(c), c)$ was made
free prior this step. This means by Lemma~\ref{l:bl}~(2) that the edge
$e(\ell(c))$ had to be made free previously. Now we recall that no triangle of
$K_{and}$ was collapsed yet (including triangles of $K_{and}$ attached to
$f(\ell)$ and $f(\bar \ell)$). Therefore, Lemma~\ref{l:literal}~(3)
implies that only one of the edges $e(\ell)$ and $e(\bar \ell)$ can be
collapsed at this stage. This gives a satisfying assignment to $\Phi$ by
setting a variable $u$ to be TRUE if $e(u)$ was collapsed (before collapsing
a triangle from $K_{and}$) and FALSE otherwise. The existence of $\ell(c)$ implies that every
clause $c$ is indeed satisfied.
\fi

\ifabstr\else
\section{Collapsing sequences}
\label{s:col_seq}
Here we prove technical lemmas used previously in the text. It is convenient to
change the order of the proofs.

\begin{proof}[Proof of Lemma~\ref{l:and}~(1).]
We recall that our task is to collapse the conjunction gadget from
Figure~\ref{f:conj_g} to the tree $A$. We start with collapsing 
the wall below the
edge $e_{and}$ and then the lowest floor (except edges belonging to $A$). 
We continue with collapsing all walls
that used to be attached to the lowest floor. At this step we have the complex depicted in
Figure~\ref{f:conj_coll}. This complex is already a 2-sphere with a hole and
with $A$ attached to it. It is easy to collapse it to $A$ in the directions of arrows. 
\end{proof}

\begin{figure}
\begin{center}
\includegraphics{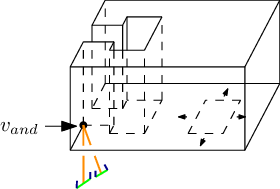}
\caption{An intermediate step in collapsing $K_{and}$.}
\label{f:conj_coll}
\end{center}
\end{figure}

\begin{proof}[Proof of Lemma~\ref{l:clause}~(1).]
We recall that our task is to collapse the clause gadget from
Figure~\ref{f:clause_g}. We will provide a collapsing to the union of paths
$p(\ell_i,c)$ and edges $(\ell_2, c)$ and $(\ell_3, c)$. Other cases are
analogous. For the picture, we will assume that Bing's room number $1$ is
above the base floor and Bing's room number $2$ is below the base floor as in
Figure~\ref{f:clause_coll} on the left.

Since $(\ell_1, c)$ is allowed to be collapsed, we can collapse the left wall
of room $2$ and then the bottom wall. In next step, we can collapse all walls
of room $2$ perpendicular to the base floor. We obtain complex as in
Figure~\ref{f:clause_coll} on the right. Next we collapse the interior of the
$23$ square so that the edges left of
$(\ell_2, c)$ become free. This means that the room $3$ can be collapsed in a
similar fashion as we collapsed room $2$ (note that after this step only
$(\ell_2, c)$, $p(\ell_2, c)$ and part of $p(\ell_1, c)$ remain of the $23$
square). Finally we can collapse room $1$ in a similar fashion taking care that
the edge $e_{and}$ remains uncollapsed.

\end{proof}

\begin{proof}[Proof of Lemma~\ref{l:literal}~(1).]
First we collapse the thick wall of $X(\bar \ell)$ in the way on
Figure~\ref{f:triang_thick}, on the right. This makes the common edge $89$ of
$X(\ell)$ and $X(\bar \ell)$ free. Then the thick wall of $X(\ell)$ can be
collapsed so that the upper Bing's room of $X(\ell)$ becomes Bing's room with
collapsed thick wall. Then the rest of $X(\ell)$ can be collapsed in very same
way as in the proof of Lemma~\ref{l:and}~(1) while keeping $p(\ell)$ and
$f(\ell)$.
\end{proof}

\begin{proof}[Proof of Lemma~\ref{l:literal}~(2).]
As soon as we are allowed to collapse $f(\bar \ell)$, the lower thick wall of
$X(\bar \ell)$ can be collapsed obtaining Bing's room with collapsed thick wall
from the lower Bing's room. Now $X(\bar \ell)$ can be collapsed in analogous
way as was presented in the proof of Lemma~\ref{l:and}~(1) while keeping the
required subcomplex (the role of the lower and upper room are interchanged). 
\end{proof}

\begin{figure}
\begin{center}
\includegraphics{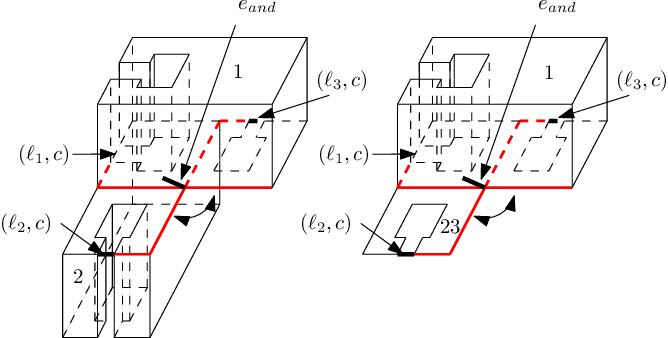}
\caption{Rooms $1$ and $2$ of the clause gadget while collapsing it.}
\label{f:clause_coll}
\end{center}
\end{figure}

\begin{proof}[Proof of Lemma~\ref{l:bl}~(1)] 
We use almost the same collapsing procedure as in the
proof of Lemma~\ref{l:and}~(1). We just remark, that the wall below $e(\ell)$,
split by $b(\ell)$ is collapsed in two stages. First the half containing
$e(\ell)$ is collapsed; then the lowest floor of $B(\ell)$ is collapsed and
finally, the second half of this wall is collapsed.
\end{proof}
\fi

\section{Conclusion}
\label{s:conclusion}

We have shown that it is NP-hard to decide whether a $3$-dimensional complex
collapses to a point. Here we mention few (simple) corollaries of our
construction as well as several related questions.

\heading{Collapsing $d$-complexes to $k$-complexes.} Motivated by a question of
Malgouyres and Franc\'{e}s~\cite{malgouyres-frances08} about higher dimensions, we set up question
$(d,k)$-\textsc{Collapsibility} asking whether a given $d$-dimensional complex collapses
to some $k$-dimensional complex where $d > k \geq 0$ are fixed parameters.

Our result shows that $(3,0)$-\textsc{Collapsibility} is NP-complete; however, it is not
difficult to observe that our result can be extended to showing that
$(d,k)$-\textsc{Collapsibility} is NP-complete for any $d \geq 3$ and $k \in \{0,1\}$.
For this it is sufficient to attach a $d$-simplex to $v_{and}$, say, and remark that
if $\Phi$ is not satisfiable, then any collapsing of $K(\Phi)$ yields a complex
of dimension $2$ or more. (We also remark that the case $d \geq 3$, $k = 1$ can
be already obtained from the construction of Malgouyres and Franc\'{e}s.)

As we mentioned in the introduction, it is not hard to see that
$(d,k)$-\textsc{Collapsibility} is polynomial time solvable whenever $d \leq 2$, and
also in the codimension 1 case (see Proposition~\ref{p:codim}).

In the remaining cases, $d \geq k + 2 \geq 2$, it is reasonable to believe that
an $NP$-hardness reduction can be obtained with higher dimensional analogues of
the gadgets in our construction. However, it does not follow from our
construction immediately, therefore we pose this case as a question.

\begin{question}
  What is the complexity status of $(d,k)$-\textsc{Collapsibility} for $d \geq k + 2 \geq
2$?
\end{question}

\heading{Collapsing to a fixed $1$-complex.}
In fact, our construction also shows that it is NP-complete 
to decide whether a $3$-dimensional complex collapses to
a fixed $1$-complex. 
For this, it is sufficient to attach the fixed $1$-complex to $v_{and}$ (and
eventually a $d$-simplex for $d \geq 3$ again if we want to reach higher
dimension).

\heading{Collapsing of complexes from a specific class.}
In general we can consider two collections of simplicial complexes, the initial
collection $\I$ and the target collection $\T$. The $(\I,
\T)$-\textsc{Collapsibility} question asks whether the given input complex from $\I$
collapses to some complex from $\T$. It would be interesting to know whether
this question is polynomial time solvable for some natural choices of $\T$ and
$\I$. One natural choice, in the author's opinion, is when $\I$ is a collection
of triangulated $d$-balls for some $d \geq 3$ and $\T$ is simply a point. Note that
even in this setting the question is non-trivial since there exist
non-collapsible $d$-balls; see, e.g.,~\cite[Corollary 4.25]{benedetti12}.
However, even in this case we suspect NP-hardness. 


\ifabstr
\else
\section*{Acknowledgements}
\fi
I would like to thank Dominique Attali, Bruno Benedetti, Anders Bj\"{o}rner,
Alexander Engstr\"{o}m, Andr\'{e} Lieutier and R\'{e}my Malgouyres
for helpful answers to my questions and/or providing useful references. I also
thank the anonymous referee for a list of very helpful comments.

\bibliographystyle{alpha}
\bibliography{../bib/general}

\begin{thebibliography}{{M}un84}

\bibitem[AL15]{attali-lieutier_online_first15}
D.~Attali and A.~Lieutier.
\newblock {Geometry-driven Collapses for Converting a \v{C}ech Complex into a
  Triangulation of a Nicely Triangulable Shape}.
\newblock {\em Discrete \& Computational Geometry}, pages 1--28, 2015.
\newblock Online first.

\bibitem[Ben12]{benedetti12}
B.~Benedetti.
\newblock Discrete {M}orse theory for manifolds with boundary.
\newblock {\em Trans. Amer. Math. Soc.}, 364(12):6631--6670, 2012.

\bibitem[BL14]{benedetti-lutz14}
B.~Benedetti and F.~H. Lutz.
\newblock Random discrete {M}orse theory and a new library of triangulations.
\newblock {\em Exp. Math.}, 23(1):66--94, 2014.

\bibitem[EG96]{egecioglu-gonzalez96}
{\"O}.~E{\u{g}}ecio{\u{g}}lu and T.~F. Gonzalez.
\newblock A computationally intractable problem on simplicial complexes.
\newblock {\em Comput. Geom.}, 6(2):85--98, 1996.

\bibitem[For98]{forman98}
R.~Forman.
\newblock {M}orse theory for cell complexes.
\newblock {\em Adv. Math.}, 134(1):90--145, 1998.

\bibitem[Hak73]{haken73}
W.~Haken.
\newblock Connections between topological and group theoretical decision
  problems.
\newblock In {\em Word problems: decision problems and the {B}urnside problem
  in group theory ({C}onf., {U}niv. {C}alifornia, {I}rvine, {C}alif., 1969;
  dedicated to {H}anna {N}eumann)}, pages 427--441. Studies in Logic and the
  Foundations of Math., Vol. 71. North-Holland, Amsterdam, 1973.

\bibitem[HAMS93]{hog-angeloni-metzler-sieradski93}
C.~Hog-Angeloni, W.~Metzler, and A.~J. Sieradski.
\newblock {\em Two-dimensional homotopy and combinatorial group theory}, volume
  197.
\newblock Cambridge University Press, 1993.

\bibitem[{H}at01]{hatcher01}
{A}. {H}atcher.
\newblock {\em {A}lgebraic Topology}.
\newblock {C}ambridge {U}niversity {P}ress, {C}ambridge, 2001.

\bibitem[JP06]{joswig-pfetsch06}
M.~Joswig and M.~E. Pfetsch.
\newblock Computing optimal {M}orse matchings.
\newblock {\em SIAM J. Discrete Math.}, 20(1):11--25 (electronic), 2006.

\bibitem[Mat03]{matousek03}
J.~Matou{\v{s}}ek.
\newblock {\em Using the {B}orsuk-{U}lam theorem}.
\newblock Universitext. Springer-Verlag, Berlin, 2003.

\bibitem[MF08]{malgouyres-frances08}
R.~Malgouyres and A.~R. Franc{\'e}s.
\newblock Determining whether a simplicial 3-complex collapses to a 1-complex
  is {NP}-complete.
\newblock {\em DGCI}, pages 177--188, 2008.

\bibitem[{M}un84]{munkres84}
{J}.~{R}. {M}unkres.
\newblock {\em {E}lements of Algebraic Topology}.
\newblock {A}ddison - {W}esley, 1984.

\bibitem[Nab95]{nabutovsky95}
A.~Nabutovsky.
\newblock {Einstein structures: Existence versus uniqueness}.
\newblock {\em Geom. Funct. Anal.}, 5(1):76--91, 1995.

\bibitem[Tan10]{tancer10np}
M.~Tancer.
\newblock $d$-collapsibility is {NP}-complete for $d$ greater or equal to 4.
\newblock {\em Chicago Journal of Theoretical Computer Science}, 2010(3):1--28,
  June 2010.

\bibitem[Tan12]{tancer12collapsibilityarxiv_v1}
M.~Tancer.
\newblock Recognition of collapsible complexes is {NP}-complete.
\newblock Preprint (version v1); http://arxiv.org/abs/1211.6254v1, 2012.

\bibitem[VKF74]{volodin-kuznetsov-fomenko74}
I.A. Volodin, V.E. Kuznetsov, and A.T. Fomenko.
\newblock {The problem of discriminating algorithmically the standard
  three-dimensional sphere}.
\newblock {\em Usp. Mat. Nauk}, 29(5):71--168, 1974.
\newblock In Russian. English translation: Russ. Math. Surv. 29,5:71--172
  (1974).

\bibitem[Weg75]{wegner75}
G.~Wegner.
\newblock $d$-collapsing and nerves of families of convex sets.
\newblock {\em Arch. Math.}, 26:317--321, 1975.

\end{thebibliography}

\ifabstr
\else 

\appendix
\section{Unrecognizability of contractible complexes}
\label{a:contractibility}

Here we briefly discuss the current state of art regarding the recognition of
contractible complexes. We first focus on the case of complexes of dimension at
least $5$. 

\begin{theorem}[Novikov]
\label{t:contractible}
For every $d \geq 5$, it is algorithmically undecidable whether a given simplicial
complex of dimension at most $d$ is contractible.
\end{theorem}

The proof of Theorem~\ref{t:contractible} easily follows from~\cite[\S 10]{volodin-kuznetsov-fomenko74}. 
We sketch a proof here. 

\begin{proof}[Sketch of a proof of Theorem~\ref{t:contractible}.]
Novikov~\cite[\S 10]{volodin-kuznetsov-fomenko74} shows
the existence of efficiently constructible sequence $M_j$ of $d$-manifolds such
that $M_j$ is a ball if and only if $\pi(M_j)$ is trivial and it is
algorithmically undecidable whether $\pi(M_j)$ is trivial. In particular, $M_j$
is contractible if and only if $\pi(M_j)$ is trivial. In order to finish the
proof, we need to know that $M_j$ can be efficiently constructed as a
simplicial complex. This can be indeed done by inspecting the
proof in~\cite{volodin-kuznetsov-fomenko74} with not too much effort.
\end{proof}

An alternative proof can be obtained from a more
complete exposition by Nabutovsky; see the appendix of~\cite{nabutovsky95}.
This is done in an earlier version of this
paper~\cite{tancer12collapsibilityarxiv_v1}. The proof there is in full detail.
It is easier to get a triangulation, because the analogues of $M_j$ are zero sets
of some polynomials with rational coefficients. On the other hand, the overall
proof is slightly more complicated since it is necessary to transform spheres into balls.

The dimension 5 in Theorem~\ref{t:contractible} can be dropped to 4, if we
greedily collapse $5$-dimensional simplices of $\Theta_i$ via some of their
$4$-dimensional faces. Note that we cannot get stuck on a $5$-dimensional
complex, since $\Theta_i$ is connected. (The author learnt this idea, in a
different context, from Bruno Benedetti.)

On the other hand, the contractibility question for complexes of dimension at most $1$ is
trivially polynomial-time solvable since it is equivalent with recognition of
trees (as graphs).  

Regarding complexes of dimension at most $2$ or $3$, the decidability of the
contractibility question is open in these two cases to the best knowledge of the author. In
particular, in dimension $2$, this question is equivalent to the triviality of
finite \emph{balanced} representations of groups; see the exercise above Subsection I.1.4 in~\cite{hog-angeloni-metzler-sieradski93}.

We also remark that it is well-known that the triviality of the fundamental group is algorithmically undecidable already for complexes of dimension 2; see,
e.g.,~\cite{haken73}. This is equivalent with contractibility of each loop in
the complex. However, this question should not be confused with the
contractibility of the complex. On the level of group presentations, the
triviality of the fundamental group corresponds to the triviality of any finite
presentations of groups, not necessarily balanced.

\fi

\end{document}